\documentclass[10pt,A4paper]{article}
\usepackage{amssymb}
\usepackage{amsmath}
\usepackage{amsthm}
\usepackage{latexsym}
\usepackage[dvips]{epsfig}
\usepackage{enumerate}
\usepackage{mathrsfs}
\usepackage{eufrak}
\usepackage{bm}
\usepackage{tikz}
\usepackage{authblk}

\theoremstyle{plain}
\newtheorem{proposition}{Proposition}
\newtheorem{lemma}{Lemma}
\newtheorem{theorem}{Theorem}
\newtheorem{assumption}{Assumption}

\newtheorem{definition}{Definition}

\def\bmg{{\bm g}}

\def\bmF{{\bm F}}

\def\bmpartial{{\bm \partial}}

\def\bari{{\bar{\iota}}}
\def\baro{{\bar{o}}}

\newcounter{mnotecount}

\newcommand{\mnotex}[1]
{\protect{\stepcounter{mnotecount}}$^{\mbox{\footnotesize $\bullet$\themnotecount}}$ 
\marginpar{
\raggedright\tiny\em
$\!\!\!\!\!\!\,\bullet$\themnotecount: #1} }

\begin{document}

\title{\textbf{Killing spinors as a characterisation of rotating black hole spacetimes}}

\author[,1,2]{Michael J. Cole \footnote{E-mail address:{\tt m.j.cole@qmul.ac.uk}}}
\author[,1]{Juan A. Valiente Kroon \footnote{E-mail address:{\tt j.a.valiente-kroon@qmul.ac.uk}}}
\affil[1]{School of Mathematical Sciences, Queen Mary, University of London,
Mile End Road, London E1 4NS, United Kingdom.}

\maketitle

\begin{abstract}
We investigate the implications of the existence of Killing spinors in
a spacetime. In particular, we show that in vacuum and electrovacuum a
Killing spinor, along with some assumptions on the associated Killing
vector in an asymptotic region, guarantees that the spacetime is
locally isometric to the Kerr or Kerr-Newman solutions. We show that
the characterisation of these spacetimes in terms of Killing spinors
is an alternative expression of characterisation results of Mars (Kerr)
and Wong (Kerr-Newman) involving restrictions on the Weyl curvature
and matter content.
\end{abstract}



\section{Introduction}

The \emph{Kerr spacetime}, describing a rotating black hole in vacuum, is one
of the most interesting exact solutions to the Einstein field
equations. As well as having physical relevance, the existence of
various incarnations of uniqueness theorems (see e.g. \cite{ChrCos12}
and references within for a survey of this vast topic) has cemented its place as one of the most
important vacuum solutions mathematically. There also exist
generalisations to spacetimes containing restricted forms of matter
---for example, the Kerr-Newman solution to the Einstein-Maxwell
equations. Although less physically interesting than the vacuum case,
these solutions still retain many interesting features of the Kerr
solution, including uniqueness under further assumptions on the matter
content. Thus, these generalisations still retain a mathematical
importance.

The remarks in the previous paragraph justify the attention given to
finding \emph{characterisations of the Kerr spacetime} and its
relations ---see e.g. \cite{Mar99,FerSae09}. Such characterisations
can be used to study various open questions about these black hole
spacetimes. For example, they can be used to reformulate uniqueness
theorems and clarify relations between them; study the stability of
the solutions, by indicating the behaviour of perturbations; and
illustrate the special characteristics of these particular solutions,
in particular through the use of symmetries ---see
e.g. \cite{AndBaeBlu15} for a recent discussion on these and related ideas. The last of these is
elegantly achieved through the use of \emph{Killing spinors}. Closely
related to \emph{Killing-Yano tensors}, these spinorial objects
represent ``hidden symmetries'' of the spacetime, which cannot be
represented using Killing vectors. It has been shown previously (see
\cite{BaeVal10b,BaeVal11b,BaeVal12}) that a vacuum spacetime admitting
a Killing spinor, along with conditions on the Weyl curvature and an
asymptotic condition, must be isometric to the Kerr spacetime. This
result crucially depends on a result of Mars (see \cite{Mar00}) which
uses the structure of the Weyl tensor, and its relation to the
Killing vectors of the spacetime, to characterise the Kerr solution in
a way that exploits to the maximum possible extent the asymptotic flatness of
the spacetime ---more precisely, it is required that the selfdual
Killing form of the stationary Killing vector is an eigenform of the
selfdual Weyl tensor.

The characterisation of the Kerr spacetime by Mars given in
\cite{Mar00} relies on a previous characterisation of this solution to
the vacuum Einstein field equations 
in terms of the vanishing of the so-called \emph{Mars-Simon tensor}
---see \cite{Mar99}. Interestingly, the latter characterisation has
been generalised to the electrovacuum case by Wong \cite{Won09}
assuming some restrictions on the matter content. This
characterisation is not optimal ---in the sense that it assumes the
existence of certain relations among the relevant geometric objects;
by contrast, in \cite{Mar99}, the existence of the vacuum counterpart
of these relations is a consequence of the
characterisation. Nevertheless, as a consequence of the analysis in
\cite{Won09}, one may expect that the Kerr-Newman solution can be
characterised by the use of Killing spinors in a similar way to the
vacuum case. The characterisations in both \cite{Mar99} and
\cite{Won09} come in both a \emph{local version} (in which certain
constraints arising in the characterisation are fixed by evaluating them at finite points of the
manifold) and a \emph{global version} (in which asymptotic flatness is
used to fix the value of the constants). Remarkably, the generalisation of the characterisation in
\cite{Mar00} to the electrovacuum case has, so far, not been obtained.

The purpose of this article is to revisit the characterisation of the
Kerr spacetime using Killing spinors and then generalise to the
electrovacuum case using Wong's result in \cite{Won09}. Our analysis
suggests that Wong's result  can be strengthened to obtain a
characterisation of the Kerr-Newman spacetime more in the spirit of
Mars's original result in \cite{Mar99} and, in turn, use this result
to obtain a generalisation of the analysis of \cite{Mar00} in which the
Kerr-Newman spacetime is characterised in an optimal way by a
combination of local and global assumptions.

\medskip
\noindent
\textbf{Outline of the article.} This paper is organised as
follows. In section \ref{KS} we give an introduction to Killing
spinors, their relation to Killing vectors and investigate the
implications on the curvature of the spacetime. We will spend some
time defining 1-forms and potentials which are useful in the
characterisations later on. In section \ref{boundary}, we define the
required asymptotic conditions needed for the characterisation
theorems. Then, in section \ref{Section:CharacteristionKerr}, we show
that the conditions of the characterisation result of Mars
\cite{Mar00} are satisfied when the spacetime admits an appropriate
Killing spinor. Finally, section \ref{Section:WongCharacterisation09}
shows the same for Wong's characterisation of the Kerr-Newman spacetime
---i.e. the existence of an appropriate Killing spinor on a
electrovacuum spacetime guarantees that the solution is Kerr-Newman up
to an isometry.

\subsection*{Conventions}
In what follows, $(\mathcal{M},\bmg)$ will denote an electrovacuum
spacetime satisfying the Einstein equations with vanishing
cosmological constant.  The signature of the metric in this article
will be $(+,-,-,-)$, to be consistent with most of the existing
literature using spinors. We use the spinorial conventions of
\cite{PenRin84}. The lowercase Latin letters $a,\, b,\, c, \ldots$ are
used as abstract spacetime tensor indices while the uppercase letters
$A,\,B,\,C,\ldots$ will serve as abstract spinor indices. The Greek
letters $\mu, \, \nu, \, \lambda,\ldots$ will be used as spacetime
coordinate indices while $\alpha,\,\beta,\,\gamma,\ldots$ will serve
as spatial coordinate indices.

Our conventions for the curvature are that
\[
\nabla_c \nabla_d u^b -\nabla_d \nabla_c u^b = R_{dca}{}^b u^a.
\]
The curvature spinors $\Psi_{ABCD}$ and $\Phi_{ABA'B'}$ are defined by
the relations
\[
\square_{AB} \xi_C = \Psi_{ABCD} \xi^D-2 \Lambda \xi_{(A}\epsilon_{B)C}, \qquad \square_{A'B'} \xi_C= \xi^D \Phi_{CDA'B'},
\]
where $\square_{AB} \equiv \nabla_{A'(A} \nabla_{B)}{}^{A'}$. Given an
antisymmetric rank 2 tensor $F_{ab}$, its Hodge dual is defined by
\[
F^\star_{ab} \equiv \frac{1}{2}\epsilon_{ab}{}^{cd}F_{cd}.
\]
The self-dual version of $F_{ab}$ is then defined by
\[
\mathcal{F}_{ab} \equiv F_{ab} + \mbox{i} F^\star{}_{ab}.
\]

\section{Killing spinors}
\label{KS}
The purpose of this section is to provide a summary of the basic
theory of Killing spinors in electrovacuum spacetimes ---see
\cite{HugSom73a,HugSom73b,HugPenSomWal72}. Throughout we assume that
$(\mathcal{M},\bmg)$ denotes an electrovacuum spacetime. In spinorial
notation the Einstein-Maxwell equations read
\[
\Phi_{ABA'B'} = 2 \phi_{AB} \bar{\phi}_{A'B'}, \qquad \Lambda =0
\]
where $\phi_{AB}=\phi_{(AB)}$ is the Maxwell spinor satisfying
\begin{equation}
\nabla^A{}_{A'}\phi_{AB}=0.
\label{MaxwellEquation}
\end{equation}
The Bianchi identity in electrovacuum spacetimes takes the form
\begin{equation}
\nabla^{A}_{\phantom{A}A'}\Psi_{ABCD}^{\phantom{A}}=2\bar{\phi}_{A'B'}\nabla_{B}^{\phantom{B}B'}\phi_{CD}^{\phantom{A}}. \label{BianchiEV}
\end{equation}

\subsection{Basic equations}
\label{Section:KillingSpinorBasicExpressions}
A \emph{Killing spinor} is a valence-2 symmetric spinor $\kappa_{AB}$
satisfying the equation
\begin{equation}
\nabla_{A'(A}\kappa_{BC)}=0.
\label{KillingSpinorEquation}
\end{equation}
By taking a further contracted derivative of this equation, it can be
shown that a solution to equation \eqref{KillingSpinorEquation} must also satisfy the integrability condition
\begin{equation}
\kappa_{(A}^{\phantom{mi}F}\Psi_{BCD)F}^{\phantom{a}}=0 \label{intcond}
\end{equation}
where $\Psi_{ABCD}$ is the Weyl spinor, a completely symmetric spinor
which is the spinorial equivalent of the Weyl tensor. This condition
restricts the form of the Weyl spinor as it requires that
\[
\Psi_{ABCD}\propto\kappa_{(AB}\kappa_{CD)}.
\]
This proportionality condition forces the spacetime to be of Petrov
type D, N or O (i.e. conformally flat). In particular, if a
non-vanishing Killing spinor has a repeated principal spinor
$\alpha_A$ so that $\kappa_{AB} = \alpha_{(A}\alpha_{B)}$, then the
Weyl spinor has four repeated null directions, and so it is of Petrov type
N.  If the Killing
spinor is algebraically general, i.e. there exist $\alpha_A$ and $\beta_B$ such that $\kappa_{AB}=\alpha_{(A}\beta_{B)}$, then the Weyl spinor has two pairs of repeated null directions, and so it is of Petrov type D.

\subsubsection*{Algebraically general Killing spinors}
 In the case that the Killing spinor $\kappa_{AB}$ is
algebraically general, we can use the principal spinors $\alpha_A$ and
$\beta_B$ to form a normalised spin dyad which we will denote by
$\{o^A, \,\iota^B\}$ and such that $o_A \iota^A=1$. The Killing spinor
$\kappa_{AB}$ is then expanded in terms of the basis as
\begin{equation}
\kappa_{AB}=\varkappa o_{(A}\iota_{B)} \label{KappaBasis}
\end{equation}
for some factor of proportionality $\varkappa$. Due to equation
\eqref{intcond}, the Weyl spinor can be expanded in a similar way as
\begin{equation}
\Psi_{ABCD}=\psi o_{(A}o_{B}\iota_{C}\iota_{D)} \label{Psibasis}
\end{equation}
for some factor of proportionality $\psi$.

The substitution of
expression \eqref{KappaBasis} in the Killing spinor equation
\eqref{KillingSpinorEquation} implies restrictions on the
Newman-Penrose (NP) spin connection coefficients. Namely, one has that
\[
\kappa = \lambda = \nu = \sigma =0,
\]
consistent with the fact that the spacetime is, at least, of Petrov
type D.

\subsection{The Killing vector associated to a Killing spinor}
\label{Section:TheKillingVector}
A Killing spinor $\kappa_{AB}$ can be used to define the spinorial
counterpart $\xi_{AA'}$ of a (possibly complex) vector via the relation
\begin{equation}
\xi_{AA'}\equiv \nabla^{C}{}_{A'}\kappa_{AC} \label{defKV}.
\end{equation}
It can be shown, using the Killing spinor equation
\eqref{KillingSpinorEquation} and commuting covariant derivatives, that
$\xi_{AA'}$ satisfies the equation
\[
\nabla_{AA'}\xi_{BB'}+\nabla_{BB'}\xi_{AA'}=-6\kappa_{(A}^{\phantom{iA}C}\Phi_{B)CA'B'}^{\phantom{A}}.
\]
Therefore, if 
\begin{equation}
\kappa_{(A}^{\phantom{iA}C}\Phi_{B)CA'B'}^{\phantom{A}}=0
\label{MatterAlignmentCondition}
\end{equation}
then $\xi_{AA'}$ is the spinorial counterpart of  a (possibly complex)
Killing vector in the spacetime. In what follows, we call condition
\eqref{MatterAlignmentCondition} the \emph{matter alignment
  condition}. In the particular case of an electrovacuum spacetime the
matter alignment condition takes the form
\begin{equation}
\kappa_{(A}{}^C \phi_{B)C}=0
\label{MaxwellAlignmentCondition}
\end{equation}
implying that the spinors $\kappa_{AB}$ and $\phi_{AB}$ are
proportional to each other. Thus, in terms of the basis dyad
$\{o,\,\iota\}$ used to express equation \eqref{KappaBasis} one can
  write
\begin{equation}
\phi_{AB} = \varphi o_{(A}\iota_{B)}
\label{PhiBasis}
\end{equation}
with $\varphi$ a proportionality constant.

\medskip
As discussed in \cite{PenRin86}, the notion of a
Lie derivative is, in general, not well defined for spinors. However,
in the case of a Hermitian spinor $\xi^{AA'}$ associated to a real
Killing vector, there exists a consistent expression which can be used
to obtain the spinorial counterpart of the condition $\mathcal{L}_\xi
F_{ab}=0$, stating that the Killing vector $\xi^a$ is also a symmetry
of the Faraday tensor ---namely:
\begin{equation}
\mathcal{L}_\xi \phi_{AB} \equiv \xi^{CC'} \nabla_{CC'} \phi_{AB} +
\phi_{C(A} \nabla_{B)C'} \xi^{CC'}.
\label{LieDMaxwell}
\end{equation}
The Maxwell spinor will be said to \emph{inherit the symmetry generated by
the Killing vector} $\xi^a$ if $\mathcal{L}_\xi \phi_{AB}=0$ ---recall
that the Maxwell spinor $\phi_{AB}$ is related to the Faraday tensor
via the relation
\[
\mathcal{F}_{AA'BB'} = 2 \phi_{AB}\epsilon_{A'B'} 
\]
where $\mathcal{F}_{AA'BB'}$ denotes the spinorial counterpart of the
self-dual Faraday tensor $\mathcal{F}_{ab}= F_{ab}+ \mbox{i}
F^\star_{ab}$. 

\medskip
\noindent
\textbf{Remark.} In Section \ref{Section:ErnstFormExpansions} it will be shown that in an
electrovacuum spacetime $(\mathcal{M},\bmg,\bmF)$ endowed with a
Killing spinor $\kappa_{AB}$ such that $\xi_{AA'}$ is Hermitian and 
$\phi_{AB}$ and $\kappa_{AB}$ satisfy the alignment condition
\eqref{MaxwellAlignmentCondition} then $\phi_{AB}$
inherits the symmetry of the spacetime. 

\subsection{Relation to Killing-Yano tensors}
If a spacetime $(\mathcal{M},\bmg)$ admits a Killing spinor
$\kappa_{AB}$, and the vector $\xi^{AA'}$ defined by \eqref{defKV} satisfies $\xi^{AA'}=\bar{\xi}^{AA'}$ (i.e. is a real vector), then one can construct a real, valence-2 antisymmetric
tensor $Y_{ab}$ as the tensorial counterpart of the spinorial relation
\begin{equation*}
 Y_{AA'BB'}\equiv \mbox{i}\big(\kappa_{AB}\epsilon_{A'B'}- \bar{\kappa}_{A'B'}\epsilon_{AB}\big)
\end{equation*}
which, as a consequence of \eqref{KillingSpinorEquation}, satisfies the Killing-Yano equation
\begin{equation*}
\nabla_{(a}Y_{b)c}=0.
\end{equation*}
Such a tensor is called a \emph{Killing-Yano tensor}. Conversely, if a
spacetime admits a Killing-Yano tensor $Y_{ab}$, one can construct a
valence-2 symmetric spinor $\kappa_{AB}$ from the relation
\begin{equation*}
\kappa_{AB}\equiv-\frac{\mbox{i}}{4}\epsilon^{A'B'}\left(Y_{AA'BB'}+\mbox{i}Y^{*}_{AA'BB'}\right)
\end{equation*}
which satisfies the Killing spinor equation
\eqref{KillingSpinorEquation} ---see e.g. \cite{PenRin86}, Section 6.7
page 107; also \cite{McLvdB93}.

\medskip
\noindent
\textbf{Remark.} The existence of a Killing-Yano tensor for the
Kerr-Newman spacetime is a key ingredient to show the integrability of
the Hamilton-Jacobi equations for geodesic motion, and the
separability of the Maxwell equations and the Dirac equation on the
Kerr-Newman spacetime ---see e.g. \cite{Kam88} or \cite{SKMHH} for further details.

\subsection{The Killing form}
\label{Section:KillingForm}
In the reminder of this section assume that the matter alignment condition
\eqref{MatterAlignmentCondition} is satisfied, so that $\xi_{AA'}$ is
the spinorial counterpart of a Killing vector. Moreover, \emph{assume
  that $\xi_{AA'}$ is a Hermitian spinor so that, in fact, it is the
  spinorial counterpart of a real vector.} Then, define the
spinorial counterpart of the \emph{Killing form} of $\xi^a$, namely
\[
H_{ab}\equiv  \nabla_{[a} \xi_{b]} =\nabla_a \xi_b
\]
by
\[
H_{AA'BB'}\equiv \nabla_{AA'}\xi_{BB'}.
\]
As a consequence of the antisymmetry in the pairs ${}_{AA'}$ and
${}_{BB'}$ it can be decomposed into irreducible parts as 
\begin{equation}
H_{AA'BB'}=\eta_{AB}\epsilon_{A'B'}+\bar{\eta}_{A'B'}\epsilon_{AB} \label{KillingForm}
\end{equation}
where $\eta_{AB}$ is a symmetric spinor ---the \emph{Killing form
  spinor}. In the sequel, we will require the self-dual part of $H_{AA'BB'}$
, denoted by $\mathcal{H}_{AA'BB'}$, and defined by
\[
\mathcal{H}_{AA'BB'} \equiv H_{AA'BB'}+\mbox{i} H_{AA'BB'}^{\star}.
\]
A direct calculation then yields 
\[
\mathcal{H}_{AA'BB'}=2\eta_{AB}\epsilon_{A'B'}.
\]
Using equation \eqref{KillingForm}, the spinor $\eta_{AB}$ can be
expressed in terms of the Killing vector as
\begin{equation}
\eta_{AB}=\frac{1}{2}\nabla_{AA'}\xi_{B}{}^{A'}.
\label{DefinitionEta}
\end{equation}
Then, by using \eqref{defKV}, this can be expanded in terms of the
Killing spinor so that 
\begin{equation}
\eta_{AB}=-\frac{3}{4}\Psi_{ABCD}\kappa^{CD}. \label{eta2kappa}
\end{equation}

\subsubsection*{Expansions for the algebraically general case}
Assuming, now that $\kappa_{AB}$ is algebraically general, using the
basis expansions of both $\kappa_{AB}$ and $\Psi_{ABCD}$, we can find
the basis expansion of $\eta_{AB}$:
\begin{equation}
\eta_{AB}=\frac{1}{4}\varkappa \psi o_{(A}\iota_{B)} = \eta
o_{(A}\iota_{B)}
\label{etabasis} 
\end{equation}
where 
\begin{equation}
\eta \equiv \frac{1}{4}\varkappa \psi.
\label{EtaScalarExplicit}
\end{equation}

\subsection{The Ernst forms and potentials}
\label{Section:ErnstPotentialsTheory}

Throughout this section let $\xi^a$ denote a real Killing vector on
the electrovacuum spacetime $(\mathcal{M},\bmg)$. A well-known
consequence of the Killing equation
\[
\nabla_a\xi_b + \nabla_b \xi_a =0
\]
and the definition of the Riemann tensor in terms of commutators of
covariant derivatives is that
\begin{equation}
\nabla_a \nabla_b \xi_c = R_{cba}{}^d \xi_d.
\label{IdentityKillingVectorRiemannBasic}
\end{equation}

\medskip
The \emph{Ernst form} of the Killing vector $\xi^a$ is defined as
\begin{equation}
\chi_a = 2 \xi^b \mathcal{H}_{ba}
\label{DefinitionErnstForm}
\end{equation}
Several properties of the Ernst form follow from the identity
\eqref{IdentityKillingVectorRiemannBasic} recast as 
\begin{equation}
\nabla_a \mathcal{H}_{bc} = \mathcal{R}_{cba}{}^d \xi_d
\label{IndentityKillingVectorRiemann}
\end{equation}
where $\mathcal{R}_{abcd}$ denotes the \emph{self-dual} Riemann
tensor. From expression \eqref{IndentityKillingVectorRiemann} it
follows, using the identity
\[ 
{}^*R_{[abc]d} = \frac{1}{3} \epsilon_{abce}R^e{}_d,
\]
that
\[
\nabla_{[a} \mathcal{H}_{bc]} =\frac{1}{3}\epsilon_{cbae} R^e{}_d \xi^d, \qquad \nabla^a \mathcal{H}_{ab} = -R_{ba}\xi^a.
\]
A further computation using the above identities and the definition of
the Ernst form, equation \eqref{DefinitionErnstForm}, yields
\begin{equation}
\nabla_a \chi_b - \nabla_b \chi_a = -2 \epsilon_{cbae} \xi^c R^e{}_d \xi^d.
\label{ExteriorDerivativeErnstForm}
\end{equation}


\subsubsection{The vacuum case}
In \emph{vacuum} $\mathcal{R}_{abcd}=\mathcal{C}_{abcd}$, where
$\mathcal{C}_{abcd}$ denotes the self-dual Weyl tensor, and so from the symmetries of the Weyl tensor one concludes that
\[
\nabla_a \chi_b - \nabla_b \chi_a =0.
\]
Consequently, in vacuum the Ernst form is closed and thus, locally exact so that there exists a
scalar, the \emph{Ernst potential} $\chi$, such that
\[
\chi_a = \nabla_a \chi.
\]

\medskip
Let now $\xi_{AA'}$ denote the (Hermitian) spinorial counterpart of
the real Killing vector $\xi^a$. If $\xi_{AA'}$ arises from a Killing
spinor through relation \eqref{defKV}, it follows from the spinor
decomposition of $\mathcal{H}_{AA'BB'}$ that the spinorial counterpart
$\chi_{AA'}$ of the Ernst form $\chi_a$ is given by
\begin{align*}
\chi_{AA'} &=4\eta_{AB}\xi^{B}_{\phantom{B}A'} \\
&=3\kappa^{CF}\Psi_{ABCF}\nabla_{DA'}\kappa^{DB}.
\end{align*}


\subsubsection{The electrovacuum case}
In the electrovacuum case the Ernst form is no longer exact
---cf. equation \eqref{ExteriorDerivativeErnstForm}. However,
if the Faraday tensor inherits the symmetry of the spacetime
---i.e. $\mathcal{L}_\xi F_{ab}=0$--- then it is possible to construct
a further 1-form, the so-called \emph{electromagnetic Ernst form},
which can be shown to be closed. In analogy to the definition in
\eqref{DefinitionErnstForm} one sets
\begin{equation}
\varsigma_a \equiv 2 \xi^b \mathcal{F}_{ba}.
\label{DefinitionEMErnstForm}
\end{equation}
A computation then shows that 
\[
\nabla_a \varsigma_b - \nabla_b \varsigma_a = 2 \mathcal{L}_\xi \mathcal{F}_{ab}.
\]
Thus, as claimed, if $\mathcal{L}_\xi F_{ab}=0$ then $\varsigma_a$ is
closed and, accordingly locally exact so that there exists a scalar,
the \emph{electromagnetic Ernst potential} $\varsigma$ such that
\[
\varsigma_a = \nabla_a \varsigma.
\]
The spinorial version of equation \eqref{DefinitionEMErnstForm} can be
readily be found to be
\[
\sigma_{AA'} = 4 \phi_{AB} \nabla^Q{}_{A'} \kappa^B{}_Q.
\]

\subsubsection{Expansions in the algebraically general case}
\label{Section:ErnstFormExpansions}
Consider now the case of an algebraically general spinor $\kappa_{AB}$
such that $\xi_{AA'}$ as given by equation \eqref{defKV} is
Hermitian. In order to find the full basis expansions of $\chi_{AA'}$
and $\varsigma_{AA'}$ 
we need to calculate the derivative of the proportionality factor
$\varkappa$. First, note the expressions for the derivatives of the
spin basis vectors in terms of the spin coefficients of the
Newman-Penrose formalism:
\begin{subequations}
\begin{align}
\nabla_{AA'}o_{B}=&-\alpha o_{A}o_{B}\bari_{A'}-\beta \iota_{A}o_{B}\baro_{A'}+\gamma o_{A}o_{B}\baro_{A'}+\epsilon \iota_{A}o_{B}\bari_{A'} \notag \\
&-\kappa\iota_{A}\iota_{B}\bari_{A'}+\rho o_{A}\iota_{B}\bari_{A'}+\sigma\iota_{A}\iota_{B}\baro_{A'}-\tau o_{A}\iota_{B}\baro_{A'}, \label{dero}\\
\nabla_{AA'}\iota_{B}=&\;\alpha o_{A}\iota_{B}\bari_{A'}+\beta\iota_{A}\iota_{B}\baro_{A'}-\gamma o_{A}\iota_{B}\baro_{A'}-\epsilon\iota_{A}\iota_{B}\bari_{A'} \notag \\
&-\lambda o_{A}o_{B}\bari_{A'}-\mu\iota_{A}o_{B}\baro_{A'}+\nu o_{A}o_{B}\baro_{A'}+\pi\iota_{A}o_{B}\bari_{A'}. \label{deri}
\end{align}
\end{subequations}
Substituting the basis expansion for the Killing spinor into the
Killing spinor equation, using expressions \eqref{dero}-\eqref{deri}
and the relation $\epsilon_{AB}=o_{A}\iota_{B}-\iota_{A}o_{B}$, we
find that
\begin{equation}
\nabla_{AA'}\varkappa=\varkappa\left(\mu o_{A}\baro_{A'}-\pi o_{A}\bari_{A'}+\tau \iota_{A}\baro_{A'}-\rho\iota_{A}\bari_{A'}\right).  \label{derK}
\end{equation}

 The expressions obtained in the previous paragraphs allow one to obtain an
expression of the Killing spinor in terms of the spin basis. A
calculation starting from the definition \eqref{defKV} readily yields
the expression
\[
\xi_{AA'} = -\frac{3}{2} \varkappa \big(\mu o_A \bar{o}_{A'} -\pi o_A
\bar{\iota}_{A'} -\tau \iota_A \bar{o}_{A'} +\rho \iota_A \bar{\iota}_{A'}\big).
\]
If $\xi_{AA'}$ is a Hermitian spinor,
i.e. $\xi_{AA'}=\bar{\xi}_{AA'}$, then the previous expression implies
\begin{equation}
\label{Hermiticy}
\bar{\mu}\bar{\varkappa} = \mu \varkappa, \qquad \bar{\tau}
\bar{\varkappa}=\varkappa\pi, \qquad \bar{\rho}\bar{\varkappa} =
\varkappa \rho.
\end{equation}

\medskip
\noindent
\textbf{The vacuum case.} Using the previous expression along with the basis expansions for
$\kappa_{AB}$ and $\Psi_{ABCD}$, \emph{in vacuum}, the Ernst form can be expanded as
\begin{equation}
\chi_{AA'}=\frac{3}{4}\varkappa ^2\psi \left(\mu o_{A}\baro_{A'}-\pi o_{A}\bari_{A'}+\tau \iota_{A}\baro_{A'}-\rho \iota_{A}\bari_{A'}\right).\label{Ernstbasis}
\end{equation}

Intuitively, one would expect it should be possible to express the
Ernst form $\chi$ in terms of the scalars $\varkappa$ and $\psi$. As
it will be seen in Section \ref{Section:CharacteristionKerr}, the
characterisation of the Kerr spacetime given by Theorem
\ref{TheoremMars00} suggests that a combination of the form
$\mathfrak{c} + \frac{3}{4}\varkappa^2 \psi$ with $\mathfrak{c}$ a constant is a suitable
candidate. In order to compute the derivative of this expression one
needs an expression for $\nabla_{AA'} \psi$. The latter can be
obtained from the vacuum Bianchi identity
\[
\nabla^{A}_{\phantom{A}A'}\Psi_{ABCD}^{\phantom{A}}=0.
\]
Substituting the basis expansion for the Weyl spinor into the above relation,
using equations \eqref{dero} and \eqref{deri}, collecting terms and
finally making use of $\epsilon_{AB}=o_{A}\iota_{B}-\iota_{A}o_{B}$
one obtains
\begin{equation}
\nabla_{AA'}\psi=-3\psi\left(\mu o_{A}\baro_{A'}-\pi
  o_{A}\bari_{A'}+\tau\iota_{A}\baro_{A'}-\rho\iota_{A}\bari_{A'}\right).  
\label{NablaPsi}
\end{equation}
Combining the latter with expression \eqref{derK} for $\nabla_{AA'}
\varkappa$ one finds that
\[
\nabla_{AA'} \bigg( \mathfrak{c} - \frac{3}{4}\varkappa^2 \psi \bigg) =\chi_{AA'}
\]
so that one can indeed set
\[
\chi = \mathfrak{c} - \frac{3}{4}\varkappa^2 \psi \qquad\mbox{for
  some} \quad \mathfrak{c}\in \mathbb{C}. 
\]
This expression can be simplified using the following observation:
combining expressions for $\nabla_{AA'}\psi$ and
$\nabla_{AA'}\varkappa$ given by equations \eqref{NablaPsi} and
\eqref{derK}, respectively, one finds that
\[
\nabla_{AA'}\left(\varkappa^3\psi\right)=0; 
\]
accordingly, the combination $\varkappa^3\psi$ is a
constant. Therefore, one has that 
\begin{equation}
\varkappa^3\psi =\mathfrak{M}
\label{VarkappaPsi:Vacuum}
\end{equation}
with $\mathfrak{M}$ a (possibly complex) constant and one has 
\begin{equation}
\chi = \mathfrak{c} - \frac{3\mathfrak{M}}{4\varkappa}.
\label{ErnstFormKillingSpinorTheory}
\end{equation}

\medskip
\noindent
\textbf{The electrovacuum case.} From the electrovacuum Bianchi identity in the
form \eqref{BianchiEV} a calculation yields
\begin{eqnarray*}
&& \nabla_{AA'}\psi=-3\left(\psi+2\varphi\bar{\varphi}\right)\mu o_{A}\baro_{A'}+3\left(\psi-2\varphi\bar{\varphi}\right)\pi o_{A}\bari_{A'} \\
&& \hspace{3cm} -3\left(\psi-2\varphi\bar{\varphi}\right)\tau\iota_{A}\baro_{A'}+3\left(\psi+2\varphi\bar{\varphi}\right)\rho
\iota_{A}\bari_{A'}. 
\end{eqnarray*}
Similarly, from the Maxwell equations \eqref{MaxwellEquation} and the
derivatives of the basis vectors given by equations \eqref{dero} and
\eqref{deri} one finds that 
\begin{equation}
\nabla_{AA'}\varphi=-2\varphi\left(\mu o_{A}\baro_{A'}-\pi o_{A}\bari_{A'}+\tau \iota_{A}\baro_{A'}-\rho \iota_{A}\bari_{A'}\right). \label{derphiEV}
\end{equation}
Thus, a further calculation using the previous expressions yields the
following explicit expression for the electromagnetic Ernst potential:
\[
\varsigma_{AA'} =3\varkappa \varphi\big( \mu o_A \bar{o}_{A'} - \pi
o_A \bar{\iota}_{A'} + \tau \iota_A \bar{o}_{A'}-\rho\iota_{A}\bari_{A'} \big). 
\]

 In the electrovacuum case, assuming an algebraically general Killing
spinor and that the Maxwell spinor and the Killing spinor satisfy the
matter alignment condition \eqref{MaxwellAlignmentCondition}, the
characterisation of the Kerr-Newman spacetime given in Theorem
\ref{Theorem:Wong09} suggests an expression for $\varsigma$ in
terms of the scalars $\varkappa$, $\psi$ and $\varphi$ ---namely
$\mathfrak{c}'-\bar{\varkappa} \bar{\psi}/2\bar{\varphi}$ with
$\mathfrak{c}'$ a constant. Combining
the above expressions one concludes that
\begin{equation}
\nabla_{AA'}\bigg(\mathfrak{c}'-\frac{\bar{\varkappa}
  \bar{\psi}}{2\bar{\varphi}} \bigg) = \varsigma_{AA'}
\label{EMErnstPotentialCheck}
\end{equation}
so that one can set
\begin{equation}
\varsigma = \mathfrak{c}' -\frac{1}{2}\frac{\bar{\varkappa} \bar{\psi}}{\bar{\varphi}} \qquad
\mbox{for some}\quad \mathfrak{c}'\in \mathbb{C}.
\label{ElectrovacuumErnstPotential}
\end{equation}
Moreover, combining expression \eqref{derK} for $\nabla_{AA'}\varkappa$
with \eqref{derphiEV} one concludes that 
\[
\nabla_{AA'}\left(\varkappa^2\varphi\right)=0
\]
meaning the combination $\varkappa^2\varphi$ is constant. Thus, there
exists a (possibly complex) constant $\mathfrak{Q}$ such that
\begin{equation}
\varkappa^2\varphi = \mathfrak{Q}. 
\label{ElectrovacuumIntegrationConstant}
\end{equation}

\medskip
In the electrovacuum case the relation between the scalars
$\varkappa$ and $\psi$ takes a more complicated form than in vacuum
---cf. equation \eqref{VarkappaPsi:Vacuum}. Given a complex constant
$\mathfrak{C}'$, a calculation using
expressions \eqref{derK}, \eqref{derphiEV} and relation
\eqref{ElectrovacuumIntegrationConstant} shows that
\begin{eqnarray*}
&& \nabla_{AA'}\left( \frac{\mathfrak{C}}{\bar{\varkappa}}+ \varkappa^3 \psi  \right) =
-\left(\frac{6|\mathfrak{Q}|^2\varkappa\mu}{\bar{\varkappa}^2} +\frac{\mathfrak{C}\bar{\mu}}{\bar{\varkappa}}\right)
o_A \bar{o}_{A'}
   -\left(\frac{6|\mathfrak{Q}|^2\varkappa\pi}{\bar{\varkappa}^2} +\frac{\mathfrak{C}\bar{\tau}}{\bar{\varkappa}}\right)
o_A \bar{\iota}_{A'} \\
&& \hspace{5cm} + \left(\frac{6|\mathfrak{Q}|^2\varkappa\tau}{\bar{\varkappa}^2} +\frac{\mathfrak{C}\bar{\pi}}{\bar{\varkappa}}\right)
\iota_A \bar{o}_{A'} + \left(\frac{6|\mathfrak{Q}|^2\varkappa\rho}{\bar{\varkappa}^2} +\frac{\mathfrak{C}\bar{\rho}}{\bar{\varkappa}}\right)
\iota_A \bar{\iota}_{A'}.
\end{eqnarray*}
If the spinor $\xi_{AA'}$ is assumed to be Hermitian, then the
previous expression reduces to
\[
\nabla_{AA'}\left( \frac{\mathfrak{C}}{\bar{\varkappa}}+ \varkappa^3
  \psi  \right) = -\frac{\varkappa (\mathfrak{C}+6
|\mathfrak{Q}|^2)}{\bar{\varkappa}^2}\big( \mu o_A
\bar{o}_{A'} + \pi o_A \bar{\iota}_{A'} - \tau \iota_A \bar{o}_{A'}  -
\rho \iota_A \bar{\iota}_{A'}\big). 
\]
Thus, if one chooses $\mathfrak{C} =- 6 |\mathfrak{Q}|^2$, then the combination
$\mathfrak{C}/\bar{\varkappa} + \varkappa^3 \psi$ is a constant
---that is, there exists $\mathfrak{M}'\in \mathbb{C}$ such that
\begin{equation}
 \varkappa^3
  \psi -\frac{6|\mathfrak{Q}|^2}{\bar{\varkappa}} = \mathfrak{M}'.
\label{MprimeDefinition}
\end{equation}
Thus, the scalar $\psi$ can be expressed solely in terms of
$\varkappa$ as
\begin{equation}
\psi = \frac{1}{\varkappa^3} \left( \mathfrak{M}' + \frac{6|\mathfrak{Q}|^2}{\bar{\varkappa}} \right).\label{massEV}
\end{equation}
Note that when the Maxwell field vanishes, then the constant $\mathfrak{Q}$ also vanishes and this equation reduces to the vacuum case given by \eqref{VarkappaPsi:Vacuum}.

\medskip
Finally, it is observed that expanding expression \eqref{LieDMaxwell} in terms of
the spinor basis $\{ o,\,\iota \}$ and using expressions
\eqref{etabasis} and \eqref{derphiEV} one concludes, after a calculation, that 
\[
\mathcal{L}_{\xi} \varphi_{AB}= 0
\]
---so that $\varphi_{AB}$ inherits the symmetry generated by the
Killing spinor $\kappa_{AB}$.

\subsection{Spacetimes with an algebraically special Killing spinor}

In this section we briefly consider electrovacuum spacetimes with an
algebraically special Killing spinor. These spacetimes will not play a
role in the remainder of this article. The reason for this is the
following result:

\begin{lemma}
Let $(\mathcal{M},\bmg)$ be a smooth electrovacuum spacetime with a
matter content satisfying the matter alignment condition and
admitting a valence-2 Killing spinor $\kappa_{AB}$ such that the
associated field $\xi^{AA'}$ is a Hermitian spinor. If $\kappa_{AB}$
is algebraically special (i.e. $\kappa_{AB}= \alpha_{A}\alpha_{B}$ for
some non-vanishing spinor $\alpha_A$) then $\xi^a=0$. 
\end{lemma}

\begin{proof}
It follows directly from the existence of a non-vanishing
algebraically special Killing spinor that the spacetime
$(\mathcal{M},\bmg)$ must be of Petrov type N ---see equation
\eqref{Psibasis}. That is, one has that
\begin{equation}
\Psi_{ABCD}=\psi \alpha_{A}\alpha_{B}\alpha_{C}\alpha_{D}
\label{Psi:PetrovN}
\end{equation}
for some function $\psi$. As the matter alignment condition holds by
assumption, the Hermitian spinor $\xi_{AA'}$ is the spinorial
counterpart of a real Killing vector $\xi^a$. The content of the
Killing form of $\xi^a$ is encoded in the symmetric spinor
$\eta_{AB}$. Substituting the expansions \eqref{Psi:PetrovN} and
$\kappa_{AB}=\alpha_A\alpha_B$ into equation \eqref{eta2kappa}, it
follows directly that $\eta_{AB}=0$. Thus, the Killing form $H_{ab}$
of $\xi^a$ vanishes. Accordingly, $\xi^a$ is a covariantly constant
vector on $(\mathcal{M},\bmg)$ ---i.e. one has
\begin{equation}
\nabla_{a}\xi^{b}=0.
\label{CovariantlyConstantKillingVector}
\end{equation}

In order to  further investigate the consequences of this observation we introduce a
normalised spin dyad $\{ o^A,\, \iota^A\}$ with $o_A=\alpha_A$ and
$o_A\iota^A=1$. One can then write
\[
\kappa_{AB}=o_A o_B, \qquad \phi_{AB}= \varphi o_A o_B.
\]
Substituting the first of the above expressions into the Killing
spinor equation $\nabla_{A'(A}\kappa_{BC)}=0$ one finds that
\begin{equation}
\gamma = \alpha = \sigma =\kappa =0, \qquad \rho+ \epsilon =0, \qquad 
\tau+\beta =0.
\label{KillingSpinorEquationAlgebraicSpecial}
\end{equation}
Moreover, one finds that the Hermitian spinor $\xi_{AA'}$ can be
expressed as
\[
\xi_{AA'} = -3 \beta o_A \bar{o}_{A'} + 3 \epsilon o_A \bar{\iota}_{A'}.
\]
The spinorial version of equation
\eqref{CovariantlyConstantKillingVector} implies $D\xi_{AA'}=0$, $\Delta
\xi_{AA'}=0$, $\delta \xi_{AA'}=0$ and $\bar{\delta} \xi_{AA'}=0$. In particular, from $\Delta \xi_{AA'}=0$ and $\bar{\delta}\xi_{AA'}=0$,
expanding in terms of the basis one finds that
$\beta\tau=0$ and $\epsilon\rho=0$. Combining this expression with the third and fourth conditions in
\eqref{KillingSpinorEquationAlgebraicSpecial} one concludes that
\[
\tau=\beta=\epsilon=\rho=0.
\]
It follows then that 
\[
\xi_{AA'} = 0.
\]

\end{proof}

As we want to use the asymptotics of the Killing vector $\xi_{AA'}$ in
the characterisation of the Kerr and Kerr-Newman spacetime, we will
rule out the algebraically special case and assume that the Killing
spinor is algebraically general ---i.e. $\kappa_{AB}\kappa^{AB}\neq0$.

\medskip
\noindent
\textbf{Remark.} Note that as we have
$\Psi_{ABCD}\propto\kappa_{(AB}\kappa_{CD)}$, then the conditions
$\Psi_{ABCD}\Psi^{ABCD}\neq0, \Psi_{ABCD}\neq0$ imply that the Killing
spinor is algebraically general and non-zero,
i.e. $\kappa_{AB}\kappa^{AB}\neq0, \kappa_{AB}\neq0$. These two
conditions on the curvature are precisely the ones assumed in Theorem
6 of \cite{BaeVal10b}, and so the characterisation of Kerr in terms of
Killing spinors presented in that article is essentially the same as
the one presented here. Despite this, we reproduce the result in this
paper for completeness and ease of comparison with the electrovacuum
case. We do this using the local result of Mars given in \cite{Mar99},
whereas the proof in \cite{BaeVal10b} uses the global result from
\cite{Mar00}. In the absence of a generalisation to the electrovacuum
case of the characterisation of \cite{Mar00}, our analysis of the
Kerr-Newman spacetime must make use of the local result by Wong
\cite{Won09}.

\section{Boundary conditions}
\label{boundary}
\label{Section:BoundaryConditions}

This section provides a brief discussion of the boundary conditions
which will be used in conjunction with the properties of Killing
spinors to characterise the Kerr and Kerr-Newman spacetimes.

\subsection{Stationary asymptotically flat ends}

In the remainder of this article we will be particularly interested in spacetimes
with a \emph{stationary asymptotically flat 4-end} ---see e.g. \cite{Won09}.

\begin{definition}
\label{Definition:StationaryEnd}
A stationary asymptotically flat 4-end in an electrovacuum spacetime $(\mathcal{M},\bmg,\bmF)$ is an open
submanifold $\mathcal{M}_\infty\subset \mathcal{M}$ diffeomorphic to
$I\times ( \mathbb{R}^3\setminus \mathcal{B}_R)$ where $I\subset
\mathbb{R}$ is an open interval and $\mathcal{B}_R$ is a closed ball
of radius $R$. In the local coordinates $(t,x^\alpha)$ defined by the
diffeomorphism the components $g_{\mu\nu}$ and $F_{\mu\nu}$ of the
metric $\bmg$ and the Faraday tensor $\bmF$ 
satisfy 
\begin{subequations}
\begin{eqnarray}
&& |g_{\mu\nu} -\eta_{\mu\nu}| +|r \partial_\alpha g_{\mu\nu} | \leq C
r^{-1}, \label{StationaryEndCondition1}\\
&& |F_{\mu\nu}| + |r \partial_\alpha F_{\mu\nu}| \leq C' r^{-2}, \label{StationaryEndCondition2}\\
&& \partial_t g_{\mu\nu}=0, \label{StationaryEndCondition3}\\
&& \partial_t F_{\mu\nu} =0, \label{StationaryEndCondition4}
\end{eqnarray}
\end{subequations}
where $C$ and $C'$ are positive constants, $r\equiv (x^1)^2 +
(x^2)^2 + (x^3)^2,$ and $\eta_{\mu\nu}$ denote the components of the
Minkowski metric in Cartesian coordinates. 
\end{definition}

\medskip
\noindent
\textbf{Remark 1.} It follows from condition
\eqref{StationaryEndCondition3} in Definition
\ref{Definition:StationaryEnd} that the stationary asymptotically flat
end $\mathcal{M}_\infty$ is endowed with a Killing vector $\xi^a$
which takes the form $\bmpartial_t$ ---a so-called \emph{time
translation}. From condition \eqref{StationaryEndCondition4} one has
that the electromagnetic field inherits the symmetry of the spacetime ---that is
$\mathcal{L}_\xi \bmF=0$, with $\mathcal{L}_\xi$ the Lie derivative
along $\xi^a$.

\medskip
Of particular interest will be those stationary asymptotically flat
ends \emph{generated by a Killing spinor}:

\begin{definition}
\label{Definition:KillingGeneratedEnd}
A stationary asymptotically flat end $\mathcal{M}_\infty\subset \mathcal{M}$ in an
electrovacuum spacetime $(\mathcal{M},\bmg,\bmF)$ endowed with a
Killing spinor $\kappa_{AB}$ is said to be generated by a
Killing spinor if the spinor $\xi_{AA'}\equiv
\nabla^B{}_{A'} \kappa_{AB}$ is the spinorial
counterpart of the Killing vector $\xi^a$. 
\end{definition}

\medskip
\noindent
\textbf{Remark 2.} Stationary spacetimes have a natural definition of mass in terms of
the Killing vector $\xi^{a}$ that generates the isometry  ---the
so-called \emph{Komar mass} $m$ defined as
\[
m\equiv
-\frac{1}{8\pi}\lim_{r\rightarrow\infty}\int_{S_{r}}\epsilon_{abcd}\nabla^c\xi^d \mbox{d}S^{ab}
\]
where $S_{r}$ is the sphere of radius r centred at $r=0$ and
$\mbox{d}S^{ab}$ is the binormal vector to $S_r$.  Similarly, one can
define the \emph{total electromagnetic charge} of the spacetime via the integral
\[
q= -\frac{1}{4\pi}\lim_{r\rightarrow\infty}\int_{S_{r}} F_{ab} \mbox{d}S^{ab}.
\]


\medskip
\noindent
\textbf{Remark 3.} In the asymptotic region the components of the metric can be written
in the form
\begin{eqnarray*}
&& g_{00} = 1- \frac{2 m}{r} + O(r^{-2}), \\
&& g_{0\alpha} = \frac{4 \epsilon_{\alpha\beta\gamma} S^\beta
  x^\gamma}{r^3} +O(r^{-3}), \\
&& g_{\alpha\beta} = -\delta_{\alpha\beta} + O(r^{-1}),
\end{eqnarray*}
where $m$ is the Komar mass of $\xi^a$ in the end
$\mathcal{M}_\infty$,  $\epsilon_{\alpha\beta\gamma}$ is the flat
rank 3 totally antisymmetric tensor and $S^\beta$ denotes a
3-dimensional tensor with constant entries. For the components of the Faraday
tensor one has that
\begin{eqnarray*}
&& F_{0\alpha} = \frac{q}{r^2} + O(r^{-3}), \\
&& F_{\alpha\beta} = O(r^{-3})
\end{eqnarray*}
---see e.g. \cite{Sim84b}. Thus, to leading order any stationary
electrovacuum spacetime is asymptotically a Kerr-Newman spacetime. 

\medskip
\noindent
\textbf{Remark 4.} In the case of the exact Kerr-Newman spacetime with
mass $m$, angular momentum $a$ and charge $q$ a NP
frame $\{l^a,\, n^a,\, m^a,\, \bar{m}^a \}$ with associated spin dyad
$\{ o^A,\, \iota^A\}$ such that the scalars $\varkappa$, $\varphi$ and
$\psi$ introduced in equations \eqref{KappaBasis}, \eqref{PhiBasis} and \eqref{Psibasis},
respectively, take the form
\begin{eqnarray*}
&& \varkappa = \frac{2}{3} (r - \mbox{i} a \cos\theta), \\
&&  \varphi = \frac{q}{(r-\mbox{i}a \cos \theta)^2}, \\
&& \psi = \frac{6}{(r -\mbox{i}a \cos\theta)^3}\left(  \frac{q^2}{r+
    \mbox{i} a \cos\theta } -m \right),
\end{eqnarray*}
where $r$ denotes the standard \emph{Boyer-Lindquist} radial
coordinate ---see \cite{AndBaeBlu15} for more details.

\subsection{Killing spinor and Killing vector asymptotics}
\label{Section:KillingVectorAsymptotics}

In general, the spinor $\xi_{AA'}$ obtained from a Killing spinor
$\kappa_{AB}$ using formula \eqref{defKV} is not Hermitian. It is, however,
well know that for the Kerr-Newman spacetime $\xi_{AA'}$ is indeed the
spinorial counterpart of a real Killing vector $\xi^a$ ---see
e.g. \cite{AndBaeBlu15}. More generally, this observation applies to
any electrovacuum spacetime with a stationary asymptotically flat
end. To see this, we first notice the following:

\begin{lemma}
Let $(\mathcal{M},\bmg,\bmF)$ be a smooth electrovacuum spacetime with
a stationary asymptotically flat end $\mathcal{M}_{\infty}$, admitting
a complex Killing vector field $\xi^{a}$. If $\xi^{a}$ tends to a time
translation at infinity in $\mathcal{M}_{\infty}$, then $\xi^{a}$ is
in fact a real vector in  $\mathcal{M}_{\infty}$.
\end{lemma}

\begin{proof}
The complex Killing vector can be written
$\xi^{a}=\xi_{1}^{a}+\mbox{i}\xi^{a}_{2}$ for two real vectors
$\xi_{1}^{a}, \xi_{2}^{a}$, which are also Killing vectors by
linearity of the Killing vector equation. As a time translation
$(\partial_{t})^{a}$ is a real vector, we have
$\xi_{1}^{a}\rightarrow(\partial_{t})^{a}$ and
$\xi_{2}^{a}\rightarrow0$ as $r\rightarrow\infty$ in the
asymptotically flat end $\mathcal{M}_{\infty}$. However, it is well
known that there are no non-trivial real Killing vectors which vanish
at infinity ---see e.g. \cite{BeiChr96,ChrOMu81}.  Therefore, $\xi_{2}^{a}=0$ on
$\mathcal{M}_\infty$, and $\xi^{a}=\xi_{1}^{a}$ is a real Killing
vector.
\end{proof} 

Therefore, by assuming that the Killing vector $\xi^a$ is
asymptotically a time translation, then we can drop the assumption
requiring its spinorial equivalent $\xi_{AA'}$ to be a Hermitian
spinor. In fact, it is possible to replace this condition on the
Killing vector with an asymptotic condition on the Killing spinor, as
described in the following proposition:

\begin{proposition}
\label{Proposition:AsymptoticBehaviourSpinors}
Let $(\mathcal{M},\bmg,\bmF)$ denote an electrovacuum spacetime with
a stationary asymptotically flat end $\mathcal{M}_\infty$ generated by
a Killing spinor $\kappa_{AB}$. Then 
the functions $\varkappa$, $\varphi$
and $\psi$ as defined by equations \eqref{KappaBasis}, \eqref{Psibasis} and
\eqref{PhiBasis} satisfy 
\begin{eqnarray*}
&& \varkappa=\frac{2}{3}r+O(1), \\
&& \varphi=\frac{q}{r^2}+O(r^{-3}), \\
&& \psi=-\frac{6m}{r^3}+O(r^{-4}).
\end{eqnarray*}
Moreover, one has that
\[
\xi^2 \equiv \xi_{AA'} \xi^{AA'} = 1 + O(r^{-1}).
\]
\end{proposition}

\begin{proof}
The analysis in \cite{Sim84b} shows that to leading order the
electrovacuum spacetime $(\mathcal{M},\bmg,\bmF)$ coincides on
$\mathcal{M}_\infty$ with the Kerr-Newman spacetime. Thus, the
expansions for the fields $\varkappa$, $\varphi$ and $\varphi$ must
coincide, to leading order with their expressions for the Kerr-Newman
spacetime ---see \cite{AndBaeBlu15}.  

\end{proof}

\section{Characterisations of the Kerr spacetime}
\label{Section:CharacteristionKerr}

The motivation behind our analysis is the following theorem by M. Mars ---see
\cite{Mar00}:

\begin{theorem}
\label{TheoremMars00}
Let $(\mathcal{M},\bmg)$ be a smooth, vacuum spacetime admitting a
    Killing vector $\xi^a$ with selfdual Killing form
    $\mathcal{H}_{ab}$. Let $\mathcal{M}$ satisfy:
\begin{itemize}
\item[(i)] there exists a non-empty region $\mathcal{M}_\bullet\subset
  \mathcal{M}$ where 
\[
\mathcal{H}^2\equiv
  \mathcal{H}_{ab}\mathcal{H}^{ab}\neq 0;
\]
\item[(ii)] The selfdual Killing form and the Weyl tensor are related
  by
\begin{equation}
\mathcal{C}_{abcd} = H \left(\mathcal{H}_{ab}\mathcal{H}_{cd}
  -\frac{1}{3}\mathcal{H}^2 \mathcal{I}_{abcd}\right)
\label{VanishingMarsSimonVacuum}
\end{equation}
where 
\[
\mathcal{I}_{abcd} \equiv\frac{1}{4}(g_{ac}g_{bd} -g_{ad}g_{bc}+
\mbox{\em i} \epsilon_{abcd})
\]
and $H$ is a complex scalar function.
\end{itemize}
Then there exist two complex constants $\mathfrak{l}$ and $\mathfrak{c}$ such that
\[
H = \frac{6}{\mathfrak{c}-\chi}, \qquad \mathcal{H}^2 = -\mathfrak{l} (\mathfrak{c}-\chi)^4.
\] 
If, in addition, $\mathfrak{c}=1$ and $\mathfrak{l}$ is real positive, then
$(\mathcal{M},\bmg)$ is locally isometric to the Kerr spacetime. 
\end{theorem}

\medskip
\noindent
\textbf{Remark 1.} It is important to emphasise that in the above
Theorem the existence of the constants $\mathfrak{c}$ and
$\mathfrak{l}$ and the functional dependence of $H$ and $\mathcal{H}^2$
with respect to $\chi$ are a consequence
of the hypotheses of the theorem ---this should be contrasted with the electrovacuum
case in which the existence of the analogue constants needs to be
assumed. 

\medskip
\noindent
\textbf{Remark 2.} A particular case of Theorem \ref{TheoremMars00}
occurs when $(\mathcal{M},\bmg)$ is \emph{a priori} assumed to have an
stationary asymptotically flat end $\mathcal{M}_\infty$ with the
Killing vector $\xi^a$ tending asymptotically to a time translation at
infinity and such that the Komar mass
associated to $\xi^a$ is non-zero. These
assumptions ensure that $\mathcal{H}^2\neq 0$ in a region of the
spacetime ---namely, in $\mathcal{M}_\infty$. Thus, one only needs to
verify condition \eqref{VanishingMarsSimonVacuum} to conclude that
\[
H = \frac{6}{1-\chi}
\]
and that the spacetime is locally isometric to the Kerr spacetime
---see Theorem 2 in \cite{Mar99}.

\medskip
\noindent
\textbf{Remark 3.} In the subsequent discussion we will make use of the spinorial
transcription of the conditions in the previous Theorem. In
particular, we notice that the content of  the combination $\mathcal{H}_{ab}\mathcal{H}_{cd}
  -\frac{1}{3}\mathcal{H}^2 \mathcal{I}_{abcd}$ can be encoded in
  terms of the spinor $\eta_{AB}$ as defined in equation \eqref{DefinitionEta} as
\[
\left(4\eta_{AB}\eta_{CD}-\frac{2}{3}\eta_{EF}\eta^{EF}(\epsilon_{AD}\epsilon_{BC}+\epsilon_{AC}\epsilon_{BD})\right)\epsilon_{A'B'}\epsilon_{C'D'} =4\eta_{(AB}\eta_{CD)}\epsilon_{A'B'}\epsilon_{C'D'}
\]
where the last expression follows from a decomposition in irreducible
terms. Thus, condition \eqref{VanishingMarsSimonVacuum} can be
concisely expressed in terms of spinors as
\begin{equation}
\Psi_{ABCD} = 2H \eta_{(AB}\eta_{CD)}.
\label{VanishingMarsSimonVacuumSpinorial}
\end{equation}
Finally, it is noticed that the condition $\mathcal{H}^2\neq 0$ can be
expressed as
\[
8\eta_{AB} \eta^{AB}\neq 0.
\]

\subsection{Killing spinors and the Mars characterisation}

In what follows we analyse the extent to which existence of a Killing
spinor on a vacuum spacetimes implies the hypotheses of the
characterisation of Kerr given in Theorem \ref{TheoremMars00}. For
bookkeeping purposes we explicitly state the assumptions to be made in
the remainder of this section:

\begin{assumption}
Let $(\mathcal{M},\bmg)$ be a smooth vacuum spacetime and let
$\mathcal{K}\subset \mathcal{M}$ such that:
\begin{itemize}
\item[(i)] on $\mathcal{K}$ there exists an algebraically general
  Killing spinor $\kappa_{AB}$;

\item[(ii)] the spinor $\xi_{AA'} \equiv \nabla^B{}_{A'} \kappa_{AB}$
  is on $\mathcal{K}$ the spinorial counterpart of a real Killing
  spinor $\xi^a$ ---i.e. $\xi_{AA'}$ is Hermitian. 

\end{itemize}
\end{assumption}

Under the above assumptions, it follows from combining the basis
expansion for $\Psi_{ABCD}$ and $\eta_{AB}$, equations \eqref{Psibasis} and
\eqref{etabasis}, respectively, that
\[
\Psi_{ABCD} = \frac{16}{\varkappa^2\psi} \eta_{(AB}\eta_{CD)}.
\]
Thus, hypothesis \emph{(ii)} of Theorem \ref{TheoremMars00} is
satisfied with
\[
H=\frac{8}{\varkappa^{2}\psi}
\]
---cf. equation \eqref{VanishingMarsSimonVacuumSpinorial}. Using the
expression for the Ernst potential predicted by the theory of Killing
spinors, equation \eqref{ErnstFormKillingSpinorTheory}, one obtains
that
\[
H = \frac{6}{\mathfrak{c}-\chi}
\]
which is precisely the form for $H$ predicted by Theorem
\ref{TheoremMars00}. From this expression one further concludes that
\[
\mathcal{H}^2 = -\frac{\mathfrak{M}}{3}\left(  \frac{4}{3\mathfrak{M}} \right)^3(\mathfrak{c}-\chi)^4.
\]
This, again, is the form predicted by Theorem
\ref{TheoremMars00}. 

\medskip
The above observations allow us to formulate the following
\emph{Killing spinor version} of Theorem \ref{TheoremMars00}:

\begin{proposition}
\label{Proposition:MarsCharacterisationKS}
Let $(\mathcal{M},\bmg)$ denote a smooth vacuum spacetime endowed with
a Killing spinor $\kappa_{AB}$ with $\kappa_{AB}\kappa^{AB}\neq 0$
such that the spinor $\xi_{AA'} \equiv \nabla{}^B{}_{A'} \kappa_{AB}$
is Hermitian. Then there exist two complex constants $\mathfrak{l}$ and $\mathfrak{c}$ such that
\[
\mathcal{H}^2 = -\mathfrak{l} (\mathfrak{c}-\chi)^4.
\] 
If, in addition, $\mathfrak{c}=1$ and $\mathfrak{l}$ is real positive, then
$(\mathcal{M},\bmg)$ is locally isometric to the Kerr spacetime. 
\end{proposition}

\subsubsection{A characterisation using asymptotic flatness}

In this subsection we simplify the previous discussion by assuming
that the set $\mathcal{K}\subset \mathcal{M}$ contains a stationary
asymptotically flat end with the Killing spinor $\kappa_{AB}$
generating the time translation Killing vector. 

\medskip
From Proposition \ref{Proposition:AsymptoticBehaviourSpinors} it
readily follows that
\[
(\mathfrak{c}-\chi)^4 = \frac{16m^4}{r^4} + O(r^{-5}).
\]
Similarly, one has, using equation \eqref{etabasis}, that 
\[
\mathcal{H}^2 = -4 \eta^2 = -\frac{4m^2}{r^4}+O(r^{-5}).
\]
Thus, by consistency with the required asymptotic behaviour of the Ernst potential,
one has to set $\mathfrak{c}=1$ and the constant $\mathfrak{l}$ in Proposition
\ref{Proposition:MarsCharacterisationKS} is given by
$\mathfrak{l}=1/4m^2$.

\medskip 
We can summarise the discussion of the previous section in the
following:

\begin{theorem}
Let $(\mathcal{M},\bmg)$ be a smooth vacuum spacetime containing a
stationary asymptotically flat end $\mathcal{M}_{\infty}$ generated by
a Killing spinor $\kappa_{AB}$. Let the Komar mass associated to the Killing vector
$\xi_{AA'}=\nabla^{B}_{\phantom{B}A'}\kappa_{AB}$ in
$\mathcal{M}_{\infty}$ be non-zero. Then, $(\mathcal{M},\bmg)$ is locally isometric to the Kerr spacetime.
\end{theorem}

\medskip
\noindent
\textbf{Remark.} As observed in \cite{AndBaeBlu15} the requirement on
the non-vanishing of the Komar mass can be replaced by an assumption
on the existence of a horizon.

\section{Characterisations of the Kerr-Newman spacetime}
\label{Section:WongCharacterisation09}

In this section we discuss characterisations of the Kerr-Newman
spacetime through Killing spinors. Our starting point is the following
result ---see \cite{Won09}:

\begin{theorem}
\label{Theorem:Wong09}
Let $(\mathcal{M},\bmg,\bmF)$ be a smooth, electrovacuum spacetime
admitting a real Killing vector $\xi^a$. Assume further that $\xi^a$
is timelike somewhere in $\mathcal{M}$ and that $F_{ab}$ is
non-null on $\mathcal{M}$ (i.e. $\mathcal{F}^2\equiv
\mathcal{F}_{ab}\mathcal{F}^{ab}\neq 0$)  and that it inherits the
symmetry of the spacetime ---i.e.
\begin{equation}
\mathcal{L}_\xi \mathcal{F}_{ab}=0.
\label{LieDerivativeFaraday}
\end{equation}
Assume, furthermore, that there exists a complex scalar $P$, a
normalisation for $\varsigma$ and a complex constant $\mathfrak{c}_1$
such that:
\begin{subequations}
\begin{eqnarray}
&& P^{-4} = -\mathfrak{c}_1^2 \mathcal{F}^2,  \label{WongCondition1}\\
&& \mathcal{H}_{ab} = -\frac{1}{2}\bar{\varsigma} \mathcal{F}_{ab}, \label{WongCondition2}\\
&& \mathcal{C}_{abcd} = 3P \left( \frac{1}{2}
   \mathcal{F}_{ab}\mathcal{H}_{cd} + \frac{1}{2}
   \mathcal{F}_{ab}\mathcal{H}_{cd} - \frac{1}{3} \mathcal{I}_{abcd}
   \mathcal{F}_{ef} \mathcal{H}^{ef}\right). \label{WongCondition3}
\end{eqnarray}
\end{subequations}
Then there exist complex constants $\mathfrak{c}_2$ and
$\mathfrak{c}_3$ such that:
\begin{subequations} 
\begin{eqnarray}
&& P= \frac{2}{\mathfrak{c}_2-\varsigma}, \label{WongConclusions1} \\
&& 4\xi^2 - |\varsigma|^2 =\mathfrak{c}_3. \label{WongConclusions2}
\end{eqnarray}
\end{subequations}
If, further, $\mathfrak{c}_2$ is such that
$\mathfrak{c}_1\bar{\mathfrak{c}}_2$ is real and $\mathfrak{c}_3$ is
such that $|\mathfrak{c}_2|^2+\mathfrak{c}_3=4$, then
$(\mathcal{M},\bmg,\bmF)$ is locally isometric to a
     Kerr-Newman spacetime. 
\end{theorem}

\noindent
\textbf{Remark 1.} As in Section \ref{Section:CharacteristionKerr}, we will make use of a
reformulation of the conditions in Theorem \ref{Theorem:Wong09} in spinorial
formalism. A direct computation shows that \eqref{WongCondition1} can
be rewritten as
\[
P^{-4} =- 8\mathfrak{c}^2_1 \phi_{AB}\phi^{AB}.
\]
Similarly, condition \eqref{WongCondition2} can be readily expressed
in terms of the spinors $\eta_{AB}$ and $\varphi_{AB}$ as
\[
\eta_{AB} = -\frac{1}{2} \bar{\varsigma} \phi_{AB},
\]
while, finally, equation
\eqref{WongCondition3} is equivalent to
\[
\Psi_{ABCD}=6P\eta_{(AB}\phi_{CD)}.
\]

\subsection{Killing spinors and Wong's characterisation}
In this section we investigate some further consequences of the
existence of Killing spinors on electrovacuum spacetimes. For easy
reference we state the assumptions to be made in what follows:

\begin{assumption}
\label{EVAssumptions}
Let $(\mathcal{M},\bmg)$ be a smooth electrovacuum spacetime and let
$\mathcal{K}\subset \mathcal{M}$ such that:
\begin{itemize}
\item[(i)] on $\mathcal{K}$ there exists an algebraically general
  Killing spinor $\kappa_{AB}$;

\item[(ii)] the spinor $\xi_{AA'} \equiv \nabla^B{}_{A'} \kappa_{AB}$
  is on $\mathcal{K}$ the spinorial counterpart of a real Killing
  spinor $\xi^a$ ---i.e. $\xi_{AA'}$ is Hermitian;

\item[(iii)] the Killing spinor $\kappa_{AB}$ and the Maxwell spinor
  $\phi_{AB}$ satisfy the alignment condition
  $\kappa_{(A}{}^{Q}\phi_{B)Q}=0$ ---that is,
  they are proportional. 

\end{itemize}
\end{assumption}

\medskip
As already discussed in Section \ref{Section:ErnstFormExpansions},
under the above assumptions it follows that $\mathcal{L}_\xi
\phi_{AB}=0$ which, in turn, implies that $\mathcal{L}_\xi
F_{ab}=0$. Thus, the electromagnetic field inherits the symmetry
generated by the Killing spinor $\kappa_{AB}$.

From the discussion in
Sections \ref{Section:KillingSpinorBasicExpressions} and \ref{Section:KillingForm} it follows that 
\[
\Psi_{ABCD} =\psi o_{(A} o_B \iota_C \iota_{D)}, \qquad \eta_{(AB}
\phi_{CD)} = \eta \varphi o_{(A} o_B \iota_C \iota_{D)}.
\]
Thus, the spinorial version of condition \eqref{WongCondition3} in
Theorem \ref{Theorem:Wong09} is satisfied with a proportionality
function $P$ given by
\[
P =  \frac{2}{3\varkappa \varphi}.
\]
Now, making use of expressions \eqref{ElectrovacuumErnstPotential},
\eqref{ElectrovacuumIntegrationConstant} and \eqref{massEV} to rewrite $P$ in terms of the
electromagnetic Ernst potential one finds that
\[
P = \frac{2}{\mathfrak{c}_2 - \varsigma}, \qquad \mathfrak{c}_2 \equiv \mathfrak{c}'-\frac{\bar{\mathfrak{M}}'}{2\bar{\mathfrak{Q}}}.
\]
Thus, under the Assumptions \ref{EVAssumptions}, hypothesis
\eqref{WongCondition3} and conclusion
\eqref{WongConclusions1}  in Theorem \ref{Theorem:Wong09} are
satisfied.

\medskip
Moreover, from the discussion in Section
\ref{Section:ErnstFormExpansions} it follows that the spinors
$\eta_{AB}$ and $\phi_{AB}$ are proportional to each other with a
proportionality function $\bar{\varsigma}$ given by
\[
\bar{\varsigma} = -\frac{\varkappa\psi}{2\varphi}.
\]
The calculations of Section \ref{Section:ErnstPotentialsTheory},
cf. equation \eqref{EMErnstPotentialCheck} in particular, show that
$\varsigma$ satisfies the properties to be expected from the
electromagnetic Ernst potential. Therefore, by setting the constant
$\mathfrak{c}'$ in the definition of $\varsigma$ given by
\eqref{ElectrovacuumErnstPotential} to zero (and thereby fixing the
normalisation of the potential), condition \eqref{WongCondition2} is
satisfied.  A similar remark holds for condition
\eqref{WongCondition1} with the constant $\mathfrak{c}_1$ given by
\[
\mathfrak{c}_1^2 = \frac{81}{64}\mathfrak{Q}^2. 
\]

In the presence of a Killing spinor, the norm $\xi^2\equiv \xi_a\xi^a$ of the
associated Killing vector is related to the electromagnetic form
$\varsigma$. To see this consider
\begin{eqnarray*}
&& \nabla_{AA'}\xi^2 = 2 \xi^{BB'} \nabla_{AA'} \xi_{BB'} \\
&& \phantom{\nabla_{AA'}\xi^2} = -2 \eta_{AB} \xi^B{}_{A'} -2
\bar{\eta}_{A'B'} \xi_{A}{}^{B'}
\end{eqnarray*}
where in the second line it has been used that 
\[
\nabla_{AA'} \xi_{BB'} = \eta_{AB} \epsilon_{A'B'} + \bar{\eta}_{A'B'} \epsilon_{AB}.
\]
As the spinors $\eta_{AB}$ and $\phi_{AB}$ are proportional to each
other, cf. the previous paragraph, one can write 
\begin{eqnarray*}
&& \nabla_{AA'}\xi^2 = \bar{\varsigma} \xi^A{}_{B'} \phi_{AB} +
\varsigma \xi_B{}^{A'} \bar{\phi}_{A'B'} \\
&& \phantom{\nabla_{AA'}\xi^2} = \frac{1}{4}\big( \bar{\varsigma}
\nabla_{BB'}\varsigma  + \varsigma \nabla_{BB'} \bar{\varsigma}\big)\\
&& \phantom{\nabla_{AA'}\xi^2} = \frac{1}{4}\nabla_{BB'} |\varsigma|^2.
\end{eqnarray*}
Thus, locally there exists a constant $\mathfrak{c}_3$ such that 
\[
4 \xi^2 - |\varsigma|^2 =\mathfrak{c}_3.
\]
Thus, conclusion \eqref{WongConclusions2} in Theorem
\ref{Theorem:Wong09} is also a consequence of the existence of a
Killing spinor. 

\medskip
We can summarise the discussion of this section with the following
Killing spinor version Theorem \eqref{Theorem:Wong09}:

\begin{proposition}
\label{Proposition:WongKillingSpinor}
Let $(\mathcal{M},\bmg,\bmF)$ denote a smooth electrovacuum spacetime satisfying the matter alignment condition,
endowed with a Killing spinor $\kappa_{AB}$ with
$\kappa_{AB}\kappa^{AB}\neq 0$ such that the spinor $\xi_{AA'} \equiv
\nabla^B{}_{A'} \kappa_{AB}$ is Hermitian. Then there exist two
constants $\mathfrak{c}_2$ and $\mathfrak{c}_3$ such that 
\[
(\mathfrak{c}_2-\varsigma)^4= -
\left(\frac{9}{8}\mathfrak{Q}\right)^2\mathcal{F}^2, \qquad 4\xi^2 -|\varsigma|^2 =\mathfrak{c}_3.
\]
If, further, $\mathfrak{c}_2$ is such that
$\bar{\mathfrak{c}}_2\mathfrak{Q}$ is real and $\mathfrak{c}_3$ is
such that $|\mathfrak{c}_2|^2+\mathfrak{c}_3=4$, then
$(\mathcal{M},\bmg,\bmF)$ is locally isometric to a
     Kerr-Newman spacetime. 
\end{proposition}

\subsubsection{A characterisation using asymptotic flatness}

In this section we assume that the domain $\mathcal{K}\subset
\mathcal{M}$ considered in the Assumptions \ref{Proposition:WongKillingSpinor} contains an
stationary asymptotic flat end with the Killing spinor $\kappa_{AB}$
generating the time translation Killing vector. We use this further
assumption to determine the values of the constants in Proposition
\ref{Proposition:WongKillingSpinor}. 

\medskip
Combining the asymptotic expansions of Proposition
\ref{Proposition:AsymptoticBehaviourSpinors} with the relation
\eqref{ElectrovacuumIntegrationConstant} one readily concludes that
\[
\mathfrak{Q} = \frac{4}{9}q\in \mathbb{R}.
\] 
Similarly, using equation \eqref{MprimeDefinition} one
concludes that
\[
\mathfrak{M}' = -\frac{16}{9}m.
\]
A further computation using equation
\eqref{ElectrovacuumErnstPotential} and
\eqref{ElectrovacuumIntegrationConstant}, respectively, show that
\[
(\mathfrak{c}_2 - \varsigma )^4 = \bigg( \mathfrak{c}_2 -\frac{2m}{q}
+O(r^{-1}) \bigg)^4, \qquad
\left(\frac{9}{8}\mathfrak{Q}\right)^2\mathcal{F}^2 = -\frac{q^4}{r^4} +O(r^{-5}).
\]
Thus, for consistency, one has to set
\[
\mathfrak{c}_2 =\frac{2m}{q}
\]
---thus, one has that $\bar{\mathfrak{c}}_2 \mathfrak{Q} \in
\mathbb{R}$. From the previous discussion it follows that $\varsigma
=2m/q +O(r^{-1})$ so that, together with $\xi^2=1+O(r^{-1})$ one
concludes that $\mathfrak{c}_3$ as defined by equation
\eqref{WongConclusions2} is given by
\[
\mathfrak{c}_3 =4\bigg( 1-\frac{m^2}{q^2}\bigg).
\]
Accordingly one has that $|\mathfrak{c}_2|^2 +\mathfrak{c}_3 =4$ as
required. 

\medskip
One can summarise the discussion of the previous paragraphs in the
following:

\begin{theorem}
Let $(\mathcal{M},\bmg,\bmF)$ be a smooth, electrovacuum spacetime satisfying the matter alignment condition,
with a stationary asymptotically flat end $\mathcal{M}_\infty$
generated by a Killing spinor $\kappa_{AB}$. Let both the Komar mass
associated to the Killing vector $\xi_{AA'}=\nabla^B{}_{A'}
\kappa_{AB}$ and the total electromagnetic charge in
$\mathcal{M}_\infty$ be non-zero. Then $(\mathcal{M},\bmg,\bmF)$ is
locally isometric to the Kerr-Newman spacetime.
\end{theorem}

\section{Applications}

The advantage of the Killing spinor characterisation of the Kerr and
Kerr-Newman solutions is that the existence of such a spinor is a
geometric condition, with only reasonable asymptotic conditions
needing to be further assumed for the results presented above. This
geometric condition can be converted into initial data for a spacelike
Cauchy surface, in a way compatible with the constraint
equations. This can then be exploited to construct a geometric
invariant for the initial data set, which parametrises the deviation
of the resulting global development of the initial data set from the
exact Kerr or Kerr-Newman solution. Various versions of this
construction analysis have been considered in
\cite{BaeVal10a,BaeVal10b,BaeVal11b,BaeVal12} for the vacuum case. A
generalisation of these constructions to the electrovacuum case is
given in \cite{ColVal16b}.

Finally we point out that the results of this article suggest that the
characterisations of the Kerr-Newman spacetime given by Wong in
\cite{Won09} can be improved to an \emph{optimal} theorem in which
condition \eqref{WongCondition1} in Theorem \ref{Theorem:Wong09} is a
consequence of the other hypothesis. An optimal result of this type is
desirable if one is to attempt to use this type of characterisations
to construct an alternative approach to the uniqueness of black
holes.



\end{document}